\documentclass[letterpaper]{article} 
\usepackage{aaai23}  
\usepackage{times}  
\usepackage{helvet}  
\usepackage{courier}  
\usepackage[hyphens]{url}  
\usepackage{graphicx} 
\usepackage{natbib}  
\usepackage{caption} 
\usepackage{algorithm}
\usepackage{algorithmic}

\usepackage{newfloat}
\usepackage{listings}
\DeclareCaptionStyle{ruled}{labelfont=normalfont,labelsep=colon,strut=off} 
\floatstyle{ruled}
\newfloat{listing}{tb}{lst}{}
\floatname{listing}{Listing}
\usepackage{graphicx}
\usepackage{amsfonts}
\usepackage{bm}
\usepackage{bbm}
\usepackage{subfig}
\usepackage{amsmath}
\usepackage{amsthm}
\usepackage{algorithm}
\usepackage{algorithmic}
\usepackage{wrapfig}
\newcommand{\mc}{\mathcal}

\newtheorem{prop}{Proposition}
\newtheorem{corollary}{Corollary}
\usepackage{balance}

\newcommand{\squishlist}{
 \begin{list}{$\bullet$}
  { \setlength{\itemsep}{0pt}
     \setlength{\parsep}{2pt}
     \setlength{\topsep}{2pt}
     \setlength{\partopsep}{0pt}
     \setlength{\leftmargin}{1em}
     \setlength{\labelwidth}{1em}
     \setlength{\labelsep}{0.5em} } }

     \newcommand{\squishend}{
  \end{list}  }

\title{Learning Individual Policies in Large Multi-agent Systems through Local Variance Minimization}

\author{
    Tanvi Verma\textsuperscript{\rm 1},
    Pradeep Varakantham\textsuperscript{\rm 2}
}
\affiliations{
    \textsuperscript{\rm 1}Institute of High Performance Computing, Agency for Science, Technology and Research, Singapore
    
    \textsuperscript{\rm 2}School of Computing and Information Systems, Singapore Management University, Singapore\\
    vermat@ihpc.a-star.edu.sg, pradeepv@smu.edu.sg
}

\usepackage{bibentry}

\begin{document}

\maketitle

\begin{abstract} 
In multi-agent systems with large number of agents, typically the contribution of each agent to the value of other agents is minimal (e.g., aggregation systems such as Uber, Deliveroo). In this paper, we consider such multi-agent systems where each agent is self-interested and takes a sequence of decisions and represent them as a Stochastic Non-atomic Congestion Game (SNCG).  We derive key properties for equilibrium solutions in SNCG model with non-atomic and also nearly non-atomic agents. With those key equilibrium properties, we provide a novel Multi-Agent Reinforcement Learning (MARL) mechanism that minimizes variance across values of agents in the same state. To demonstrate the utility of this new mechanism, we provide detailed results on a real-world taxi dataset and also a generic simulator for aggregation systems. We show that our approach reduces the variance in revenues earned by taxi drivers, while still providing higher joint revenues than leading approaches. 

\end{abstract}


\section{Introduction}

In this paper, we are motivated by large scale multi-agent systems such as car matching systems (Uber, Lyft, Deliveroo etc.), food delivery systems (Deliveroo, Ubereats, Foodpanda, DoorDash etc.), grocery delivery systems (AmazonFresh, Deliv, RedMart etc) and traffic routing guidance. The individual agents need to make sequential decisions (e.g.,on where to position themselves) in the presence of uncertainty.  The objective in aggregation systems is typically to maximize overall revenue and this can result in individual agent revenues getting sacrificed. Towards ensuring fair outcomes for individual agents, our focus in this paper is on enabling computation of approximate equilibrium solutions. 

Competitive Multi-Agent Reinforcement Learning (MARL)~\cite{littman1994markov} with an objective of computing equilibrium is an ideal model for representing the underlying decision problem.  There has been a significant amount of research on learning equilibrium policies in MARL problems \cite{littman1994markov,watkins1992q,hu1998multiagent,littman2001friend}. However, these approaches typically can scale to problems with few agents.  Recently, there have been Deep Learning based MARL methods \cite{yu2021surprising,heinrich2016deep,pmlrv80yang18d,lowe2017multi} that can scale to a large number of agents on account of using decentralized learning. While  decentralized learning is scalable, such approaches consider all other agents in aggregate or as part of the environment and this can have an impact on overall effectiveness (as demonstrated in our experimental results). 


To that end, we propose a Stochastic Non-atomic Congestion Game (SNCG) model, which combines Stochastic Game~\cite{shapley1953stochastic} and Non-atomic Congestion Game (NCG) models \cite{roughgarden2002bad,krichene2015online}. SNCG helps exploit anonymity in interactions and minimal contribution of individual agents, which are key facets of aggregation systems of interest in this paper. More importantly, we use a key theoretical property of the SNCG model to design a novel model free multi-agent reinforcement learning approach that is suitable for both non-atomic and nearly non-atomic agents. Specifically, we prove that at equilibrium  values of non-atomic agents, which are all in the same local state are equal. Our model free MARL approach, referred to as LVMQ (Local Variance Minimizing Q-learning) reduces variance in ``best response"  agent values in local states to move joint solutions towards equilibrium solutions. 



We demonstrate the utility of our new approach for individual drivers on a large dataset from a major Asian city by comparing against leading decentralized learning approaches including Neural Fictitious Self Play (NFSP)~\cite{heinrich2016deep}, Mean Field Q-learning (MFQ) \cite{pmlrv80yang18d} and Multi-Agent Actor Critic (MAAC) \cite{lowe2017multi}.
In addition, we also provide results on a general purpose simulator for aggregation systems. 

\section{Related work}
The most relevant research is on computing equilibrium policies in MARL problems, which is represented as learning in stochastic games~\cite{shapley1953stochastic}. 
Nash-Q learning \cite{hu1998multiagent} is a popular algorithm that extends the classic single agent Q-learning \cite{watkins1992q} to general sum stochastic games. 
 In fictitious self play (FSP) \cite{heinrich2015fictitious} agents learn best response through self play. FSP is a learning framework that implements fictitious play \cite{brown1951iterative} in a sample-based fashion. There are a few works which assume existence of a potential function and either need to know the potential function a priori \cite{marden2012state} or estimate potential function's value \cite{mguni2020stochastic}. Unfortunately, all these algorithms are generally suited for a few agents and do not scale if the number of agents is very large, which is the case for the problems of interest in this paper.   

Recently, some deep learning based algorithms have been proposed to learn approximate Nash equilibrium in a decentralized manner. NFSP \cite{heinrich2016deep} combines FSP with a neural network to provide a decentralized learning approach. Due to decentralization, NFSP is scalable and can work on problems with many agents. \textit{Centralized learning decentralized execution} (CTDE) algorithms learn a centralized joint action value function during the training phase, however individual agents optimize their actions based on their local observation. MAAC \cite{lowe2017multi} is a CTDE algorithm which is a best response based learning method and does not focus on learning equilibrium policy. M3DDPG \cite{li2019robust} learn policies against altering adversarial policies by optimizing a minimax objective. SePS \cite{christianos2021scaling} is another CTDE algorithm which focuses on parameter sharing to increase scalabilty of MARL. However these methods do not consider infinitesimal contribution of individual agents. As we demonstrate in our experimental results, our approach that benefits from exploiting key properties of SNCG is able to outperform NFSP and MAAC with respect to the quality of approximate equilibrium solutions on multiple benchmark problem domains from literature.
There are some CTDE algorithms \cite{sunehag2018value,rashid2018qmix,yu2021surprising} which focus on cooperative settings and not suitable for non-cooperative setting of our interest.

Treating large multi-agent systems as mean field game (MFG) \cite{huang2003individual} is a popular approach to solve large scale MARL problems where each agent learns to play against mean field distribution of all other agents. For MFQ \cite{pmlrv80yang18d}, individual agents learn Q-values of its interaction with average action of its neighbour agents. \cite{huttenrauch2019deep} used mean feature embeddings as state representations to encode the current distribution of the agents. \cite{subramanian2020multi} extended MFQ to multiple types of players. In the context of finite MFGs, \cite{cui2021approximately} proposed approximate solution approaches using entropy regularization and Boltzmann policies. Recently \cite{subramanian2022decentralized} proposed decentralized MFG where agents model the mean field during the training process. This thread of research is closely related to our work and we show that by utilizing extra information shared by the central agent, LVMQ outperforms MFQ.


We have proposed to minimize variance in the values of agents present in same local state to learn equilibrium policy. Some works in single agent reinforcement learning have focused on direct and indirect \cite{sherstan2018comparing,jain2021variance, tamar2016learning} approaches to compute variance in the \textit{return} of a single agent. This is fundamentally different from LVMQ, where variance is across \textit{values} of multiple agents present in a local state (instead of variance in value distribution for a single agent in those approaches).

\begin{figure}
  \centering
    \subfloat[Routing network \label{routing}]{\includegraphics[scale=0.2]{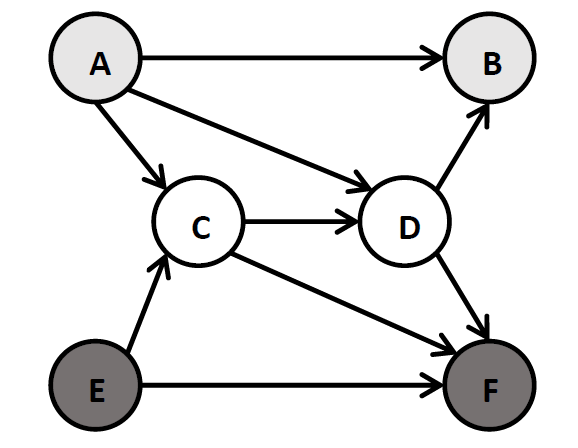}} 
     \subfloat[Grid world \label{grid-world}]{\includegraphics[scale=0.3]{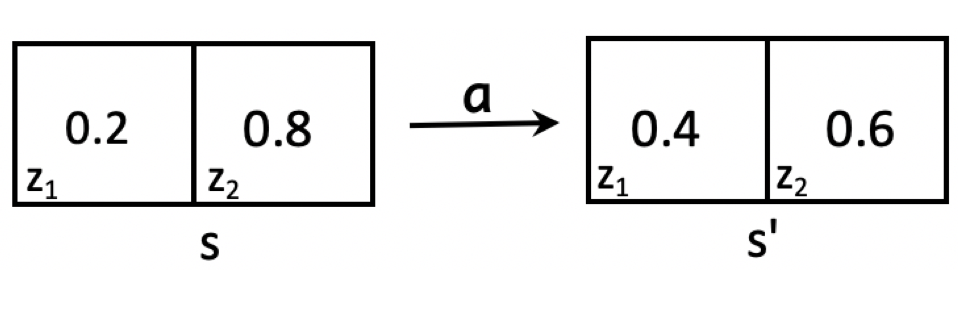}} 
  \caption{(a) Example of a routing network. (b) State transition in SNCG for a grid-world domain}
\end{figure}

\section{Background: Non-atomic Congestion Game (NCG)} \label{bg}

NCG has either been used to model selfish routing~\cite{roughgarden2002bad,roughgarden2007routing,fotakis2009structure} or resource sharing~\cite{chau2003price,krichene2015online,bilancini2016strict} problems. 
Here we present a brief overview of NCG~\cite{krichene2015online} from the perspective of resource sharing problem as that is of relevance to domains of interest, where taxis have to share demand in different zones. 

In NCG, a finite set of resources $\mc{L}$ are shared by a set of players $\mc{N}$. To capture the infinitesimal contribution of each agent, the set $\mc{N}$ is endowed with a measure space: $(\mc{N}, \mc{M}, m)$. $\mc{M}$ is a $\sigma$-algebra of measurable subsets, $m$ is a finite Lebesgue measure and is interpreted as the mass of the agents. This measure is non-atomic, i.e., for a single agent set $\{i\}$, $m(\{i\}) = 0$. The set $\mc{N}$ \big(with $m(\mc{N})$=1, i.e., mass of all agents is 1\big) is partitioned into $Z$  populations, $\mc{N} = \mc{N}_1 \cup...\cup \mc{N}_{Z}$. Each population type $z$ possesses a set of strategies $\mc{A}_z$, and each strategy corresponds to a subset of the resources. Each agent selects a strategy, which leads to a joint strategy distribution, $\bm{a} = ((f^a_z)_{a \in \mc{A}_z})_{z \in \mc{Z}} \text{ such that } \sum_{a \in \mc{A}_z} f^a_z = m(\mc{N}_z) \forall z$. Here $f^a_z$ is the total mass of the agents from population $z$ who choose strategy $a$. The total consumption of a resource $l \in \mc{L}$ in a strategy distribution $\bm{a}$ is given by: $\phi^l(\bm{a}) = \sum_{z=1}^{Z} \sum_{a \in \mc{A}_z | l \in a} f^a_z$

The cost of using a resource $l \in \mc{L}$ for strategy $\bm{a}$ is: $c_l(\phi^l(\bm{a}))$ where the function $c_l(.)$ represents cost of congestion and is assumed to be a non-decreasing continuous function. The cost experienced by an agent of type $z$ which selects strategy $a \in \mc{A}_z$ is given by: ${\mc C}^a_z(\bm{a}) = \sum_{l \in a} c_l(\phi^l(\bm{a}))$. A strategy $\bm{a}$ is a Nash equilibrium if:
$$\forall z , \forall a, a' \in \mc{A}_z: \textbf{ if } f^a_z > 0, \textbf{ then } {\mc C}^{a'}_{z}(\bm{a}) \ge {\mc C}^a_z(\bm{a})$$ Intuitively, it implies that the cost for any other strategy, $a'$ will be greater than or equal to the cost of strategy, $a$. It also implies that for a population $\mc{N}_z$, all the strategies with non-zero mass will have equal costs.

Packet routing problem in Figure \ref{routing} is an example of an NCG. Two agent populations ($Z = 2$) $\mc{N}_1$ and $\mc{N}_2$ of mass 0.5 each share the network. Each edge is treated as a resource. Agents from $\mc{N}_1$ send packets from node $A$ to node $B$, and the the agents from $\mc{N}_2$ send from node $E$ to node $F$.  Paths (strategies) $AB, ACDB, ADB$ are available for population $\mc{N}_1$ whereas paths $EF, ECDF, ECF$ are available for population $\mc{N}_2$. A joint strategy is said to be an equilibrium strategy if for population type $z$, costs on paths available to them with non-zero mass are equal. Refer to the equilibrium policy given in Table \ref{eq-policy}, the costs on paths $ACDB$ and $ADB$ (paths with non-zero mass) for population $\mc{N}_1$ will be equal at equilibrium.
\section{Stochastic Non-atomic Congestion Games (SNCG)}
\label{sncg}
In this section we propose SNCG model, which is a combination of NCG and Stochastic Games. NCG is primarily for single stage settings, we extend it to multi-stage settings in SNCG. Formally, SNCG is the tuple: $\big< \mc{N, Z, S, A, T, R} \big>$
\squishlist
\item[$\mc{N}$:] set of agents with properties similar to the ones in NCG described in Section \ref{bg}.
\item[$\mc{Z}$:]  set of local states for individual agents (e.g., location zone of a taxi).
\item[$\mc{S}$:] set of global states (e.g., distributions of taxis in the city).  At any given time, $\mc{N}$ is divided into $|\mc{Z}|$ disjoint sets, $\mc{N} = \mc{N}^s_1 \cup ... \cup \mc{N}^s_{|\mc{Z}|}$, where $\mc{N}^s_z$ is the set of agents present in local state $z$ in global state $s$ and $m(\mc{N}^s_z)$ is the mass of agents in $\mc{N}^s_z$. The distribution of mass of agents is considered as the global state, i.e., $s=<m(\mc{N}^s_1), m(\mc{N}^s_2),...,m(\mc{N}^s_{|\mc{Z}|})>, \text{ with } \sum_{z=1}^{|\mc{Z}|} m(\mc{N}^s_z) = 1 , \forall s \in \mc{S}$. 
The global state is continuous and total mass of agents in a global state is 1. 
\item[$\mc{A}$:]  is the set of actions, where $\mc{A}_z$ represents  the set of actions (e.g., locations to move to) available to individual agents in the local state $z$.  $\mc{A}= \{\mc{A}_z \}_{z \in \mc{Z}}$. Let $act(i)$ provides the action selected by agent $i$. We define $f^{a}_{z}(s) $ as the total mass of agents in $\mc{N}^s_z$ selecting action $a$ in state $s$, i.e. $\sum_{a \in \mc{A}_z} f^{a}_{z}(s)  = m(\mc{N}_z^s)$. If the agents are playing deterministic policies, $f^{a}_{z}(s)$ is given by $f^{a}_{z}(s) = \int_{i \in \mc{N}^s_z} \mathbbm{1}_{(act(i)=a)} dm(i)$. 
Similar to NCG, joint action is the distribution of masses of agents selecting available actions in each local state, i.e. for state $s$, $\bm{a} = ((f^a_z)_{a \in \mc{A}_z})_{z \in \mc{Z}} \text{ such that } \sum_{a \in \mc{A}_z} f^a_z = m(\mc{N}^s_z) \forall z$.
 
\item[$\mc{R}$]: is the reward function\footnote{``cost" is often used in context of NCG. To be consistent with MARL literature we use the term ``reward".}, which can be assumed as negative of cost function of NCGs.  The total mass of agents selecting action $a$ for a joint action $\bm{a}$ in state $s$ is given by $\phi^{a}(s, \bm{a}) =\sum_{z=1}^{|\mc{Z}|} f^{a}_{z}(s)$. $\phi(s,\bm{a})$ refers to the vector of mass of agents selecting different actions. 
Similar to the cost functions in NCG, the reward function (-ve of cost) is assumed to be a non-increasing continuous function. 

Depending on whether the congestion being represented is in source state of agent (node) or on the edge between two states, we have two possible reward cases of relevance in aggregation systems:
\squishlist
\item[\textbf{R1}]: Reward dependent only on joint state, local state and joint action, i.e., $R_z(s, \phi(s,\mathbf{a}))$. An example of this type of reward function would be referring to congestion at attractions in theme parks and individuals (agents) are minimizing their overall wait times.
\item[\textbf{R2}]: Reward dependent on joint state, local state, joint action and local action, i.e., $R_z(s, \phi^a(s, \mathbf{a}))$. This is useful in representing congestion related to traffic on roads. An example would be congestion on roads, where individual drivers are trying to reach their destination in the least time. 
\squishend

\item[$\mathcal{T}$]: $\mathcal{T}(s'|s,\bm{a})$ is the transitional probability of global states\footnote{Considering local transitions is simpler, as global transition can be computed from local transitions.} given joint actions. Similar to reward function, transition of global state is also dependent on the total mass of agents selecting different actions in state $s$. 

\squishend
In SNCG, each state is an instance of NCG where agents in a local state select actions available to them and the environment moves to the next state (another instance of NCG). Also, as the same set of actions are available to all the agents present in a given local state, agents present in the local state can be considered as agents of same population type. Figure \ref{grid-world} represents state transition from global state $s$ to $s'$ for a simple grid-world with 2 zones $z_1$ and $z_2$ and total mass of agents is 1. The number in the grid represents mass of agents present in that zone, i.e. $s=<0.2,0.8>$ and $s' = <0.4, 0.6>$. Each zone has two feasible actions, $a^{st}$ = stay in the current zone and $a^{mv}$ = move to the next zone, i.e. $\mc{A}_1, \mc{A}_2 = \{a^{st}, a^{mv}\}$. From NCG point of view, state $s$ is an NCG with 0.2 mass of agents of population type $z_1$ and 0.8 mass of agents of population type $z_2$. Similarly, $s'$ can be viewed as an NCG with 0.4 mass of agents of type $z_1$ and 0.6 mass of agents of type $z_2$.  Suppose 0.1 mass of agents in $z_1$ chose to stay and remaining chose to move to $z_2$. Similarly, 0.48 mass of agents in $z_2$ selected to stay whereas remaining agents selected action $a^{mv}$ and due to transition uncertainty, the environment moves to $s'$. The joint action (which is continuous) of the agents is given by $\bm{a} = ((0.1, 0.1),(0.48, 0.32))$. 

The policy of agent $i$ is denoted by $\pi_i$. We observe  that given a global state $s$, an agent will play different policies based on its local state $z$ as the available actions for local states are different. Hence, $\pi_i$ can be represented as $
\pi_i = (\pi_{iz}(s))_{s \in \mc{S}, z \in \mc{Z}}  \text{  such that} \sum_{a \in \mc{A}_z} \pi_{iz}(a|s) = 1$. 
We define $\Pi_z$ as the set of policies available to an agent in local state $z$, hence, $\pi_{iz}(s) \in \Pi_z  \forall i \in \mc{N}^s_z, \forall s \in \mc{S}$. $\bm{\pi} = (\pi_i)_{(i \in \mc{N})}$ is the joint policy of all the agents. 


Let $\gamma$ be the discount factor. 
The value of agent $j$ for being in local state $z$ given the global state is $s$ and other agents are following policy $\bm{\pi}_{-j}$ is given by (for the \textbf{R2} reward case):
{\small
 \begin{align}
 v_{jz}(s, \pi_{jz}, &\bm{\pi}_{-j}) = \mc{R}_z(s, \phi^{\pi_{jz}(s)}(s, \bm{a})) \nonumber \\
  &+ \gamma \int_{s'}  \mc{T}(s'| s, \bm{a})  v_{jz'}(s', \pi_{jz'}, \bm{\pi}_{-j})ds' \label{eq:val}
\end{align}
}
The goal is to compute an equilibrium joint strategy, where no individual agent has an incentive to unilaterally deviate.

\section{Properties of equilibrium in SNCG} \label{eq-sncg}

In this section, we show that values of agents in the same local state have either equal or close to equal values at equilibrium. We begin with the case of non-atomic agents and then move to the case with large number of agents and where mass of an agent is non-zero. Due to space constraints, we only provide proof sketches and the detailed proofs are provided in the appendix. 

\subsubsection{Non-Atomic Agents:}
In case of non-atomic agents, we first prove that values of other agents do not change if only one agent changes its policy (Proposition~\ref{agent-value}). This property is then later used to prove that values of agents present in a local state are equal at equilibrium (Proposition~\ref{equal-value}). 

\begin{prop} \label{agent-value}
Values of other agents do not change if agent $i$ alone changes its policy from $\pi_i$ to $\pi'_i$, i.e., for any agent $j$ in any local state $z$:
$ v_{jz} (s, \pi_{jz}, \bm{\pi}_{-j} ) = v_{jz} (s, \pi_{jz}, \bm{\pi}'_{-j} )
\text{ where }  \bm{\pi}_{-j} = \big(\pi_i, (\pi_k)_{k \in {\mc{N} \setminus \{i, j\}}}\big) $ and $\bm{\pi}'_{-j} = \big(\pi'_i, (\pi_k)_{k \in {\mc{N} \setminus \{i, j\}}}\big)$
\end{prop}
\noindent \textbf{Proof Sketch.}
As can be observed in Equation~\ref{eq:val}, change of policy for agent $i$ can have impact on agent $j$'s value, $v_{jz}(s, \pi_{jz}, \bm{\pi}_{-j})$, due to one common summary factor, which is the mass of agents taking a certain action ${a}$ in state $s$ with zone $z$, $\tilde{f}_{{z}}^{a}(s)$.

Since $f$ is primarily mass of agents, which is a Lebesgue measure and using the \textit{countable additivity} property of Lebesgue measure \cite{bogachev2007measure,hartman2014theory}, we have:
{\small  \begin{align*}
 \tilde{f}^{a}_{z}(s)  &= \int_{k \in {\mc{N}^s_z}} \mathbbm{1}_{(act(k)=a)} dm(k) -   \int_{k \in \{i\}} \mathbbm{1}_{(act(k)=a)} dm(k)
 \end{align*}
 }

\noindent Integral at a point in continuous space is 0 and mass measure is non-atomic, so $\{i\}$ is a null set and $m(\{i\})=0$. $\hfill \qedsymbol$\\

Based on Proposition~\ref{agent-value}, we now show that at Nash Equilibrium for SNCG, values of agents that are in same local state are equal.  A joint policy $\bm{\pi}$ is a Nash equilibrium if for all $z \in \mc{Z}$ and for all $i \in \mc{N}^s_z$, there is no incentive for anyone to deviate unilaterally, i.e. $
v_{iz}(s, \pi_{iz}, \bm{\pi}_{-i}) \ge v_{iz}(s, \pi'_{iz}, \bm{\pi}_{-i}) \forall s \in \mc{S}, \forall i \in \mc{N}^s_z, \forall z \in \mc{Z}, \forall \pi_{iz}, \pi'_{iz} \in \Pi_z \label{equilibrium-local}
$

\begin{prop} \label{equal-value}
Values of agents present in a local state are equal at equilibrium (denoted by *), i.e.,
{\small \begin{align*}
v_{iz}(s, \pi^*_{iz}, \bm{\pi}^*_{-i}) = v_{jz}(s, \pi^*_{jz}, \bm{\pi}^*_{-j}), \forall s \in \mc{S}, \forall i,j \in \mc{N}^s_z, \forall z \in \mc{Z}
\end{align*}}
\end{prop}
\textbf{Proof Sketch.} 
In the proof of Proposition~\ref{agent-value}, we showed that adding or subtracting one agent from a local state does not change other agent's values, as contribution of one agent is infinitesimal. Combining this with the equilibrium condition employed in NashQ learning~\cite{hu2003nash} with local states, we have: 
\begin{align}
 &v_{iz}(s, \pi^{*}_{iz}, \bm{\pi^*}) \geq v_{iz}(s, \pi^*_{jz}, \bm{\pi^*}) \quad \text{and} \nonumber \\
 &v_{jz}(s, \pi^{*}_{jz}, \bm{\pi^*}) \geq v_{jz}(s, \pi^*_{iz}, \bm{\pi^*}) \nonumber \quad \quad \quad \quad \hfill \qedsymbol
 \end{align}

These results can also be extended to problems with multiple types of agents. 

\begin{corollary}
For problems with multiple types of agents, values of same type of agents are equal in a local state at equilibrium for non-atomic case. $\qedsymbol$
\end{corollary}

In non-atomic case, individual agents have zero mass and we have shown here that values of agents with same local states will be equal at equilibrium. We now move  onto domains with large number of agents but not completely non-atomic. Since agents have non-zero mass, the proofs above do not directly hold. 

\subsubsection{Nearly Non-Atomic Agents:}
In aggregation systems, there are many agents but each agent has a small amount of mass. For this case of nearly non-atomic agents, we are only able to provide the proof for reward setting \textbf{R1}. 
\begin{prop}
When agents have non-zero mass in SNCG, consider two agents, $i$ and $j$ in zone $z$ and let $\bm{\pi}^* = (\pi_1^*,\cdots,\pi_i^*, \cdots, \pi_j^*, \cdots)$ be the equilibrium policy. For \textbf{R1} setting, we have:
$$ v_{iz}(s, \pi_{i}^*, \bm{\pi}^*_{-i}) = v_{jz}(s, \pi_{j}^*, \bm{\pi}^*_{-j})$$
\end{prop}
\noindent \textbf{Proof Sketch.}
Since the joint policy, $\bm{\pi}^*$ is the same for both $i$ and $j$, $\bm{a}$ and consequently $\phi(s,\bm{a})$ are the same for both $v_{iz}(s, \pi_{i}^*, \bm{\pi}^*_{-i})$ and $v_{jz}(s, \pi_{j}^*, \bm{\pi}^*_{-j})$ . Therefore, $R_z(s,\phi(s,\bm{a}))$ is the same as it is defined on zone and not on individual action. Transition function is on joint state and action space, so we can recursively show that future values are also the same. $\hfill \qedsymbol$

For the reward setting \textbf{R2}, we are only able to provide a trivial theoretical upper bound on the difference $v_{iz}(s, \pi_{i}^*, \bm{\pi}^*_{-i}) - v_{jz}(s, \pi_{j}^*, \bm{\pi}^*_{-j})$ as shown in the appendix. Therefore, we empirically evaluate our hypothesis that values of agents in same state are equal with reward setting \textbf{R2}.

SNCG model requires transition and reward models, which are typically hard to have  a priori. Hence, using the insight of Proposition~\ref{equal-value}, we provide a model free deep multi-agent learning approach to compute approximate equilibrium joint policies in multi-agent systems where there are large (yet finite) number of agents. 
We demonstrate in our experimental results, our solutions are better (both in terms of overall solution quality and reduction in revenue variance of all agents) than the leading approaches for MARL problems on multiple benchmark problems.

\section{Local Variance minimizing Multi-agent Q-learning, LVMQ} \label{sec-vmq}
Given the propositions in section \ref{eq-sncg} for SNCGs, we hypothesize that values of all the agents present in any local state are equal at equilibrium (even for reward case \textbf{R2}) when there are a large number of agents with minor contributions from each agent (nearly non-atomic agents). This translates to \textbf{having zero variance\footnote{We have used zero variance as an indicative of equal values of agents in the paper, however, other measures such as standard deviation can also be used for the same purpose.} in individual agent values at each local state, while individual agents are playing their best responses}. Please note that a joint policy which yields equal values in local states is an equilibrium policy only if agents are also playing their best response strategy. This is the key idea of our LVMQ approach and is an ideal insight for approximating equilibrium solutions in aggregation systems. The centralized entity can focus on ensuring values of agents in same local states are (close to) equal by estimating a joint policy which minimizes variance in values, while the individual suppliers can focus on computing best response solutions. During learning phase, the role of the
\squishlist
\item Central entity is to ensure that the exploration of individual agents moves towards a joint policy where the variance in values of agents in a local state is minimum.  
\item Individual agents is to learn their best responses to the historical behaviour of the other agents based on guidance from central entity. 
\squishend

We now describe the four main steps of the algorithm.\\
\noindent \textbf{\textit{Central agent suggests local variance minimizing joint action, $\mathbf{a}^c$ to individual agents:}}
We employ policy gradient framework for central agent to learn a joint policy which minimizes variance in values of individual agents in local states. 
Learning experience of the central agent is given by $<s,\bm{a},\nu>$, where $\nu$ is the mean of variances in the current values of agents in local states, i.e., $\nu = \dfrac{1}{|\mc{Z}|} \sum_{z \in \mc{Z}} \nu_z$. $\nu_z$ represents the variance in values of agents in $\mc{N}^s_z$. We define two parameterized functions: joint policy function $\mu(s; \theta_{\mu})$ (joint policy which yields minimum variance in individual values) and variance function $\sigma(s, \bm{a}; \theta_\sigma)$. Since the goal is to minimize variance, we need to update joint policy parameters in the negative direction of the gradient of $\sigma(s, \bm{a})$, i.e., $-\nabla_{\theta_{\mu}}\sigma(s, \mu(s;\theta_{\mu}); \theta_\sigma)$. The gradient for policy parameters $\theta_{\mu}$ can be computed using chain rule as follows
{\small \begin{align}
-\nabla_{\theta_{\mu}}\sigma(s, \mu(s;\theta_{\mu}); \theta_\sigma) = -\nabla_{\theta_{\mu}}\mu(s;\theta_{\mu})\cdot \nabla_{\bm{a}}\sigma(s, \bm{a}; \theta_\sigma)|_{\bm{a} = \mu(s;\theta_{\mu})} \label{gradient}
\end{align}
}
\noindent \textbf{\textit{Individual agents play suggested or best response action:}} Individual agents either follow the suggested action with $\epsilon_1$ probability or play their best response policy with $1-\epsilon_1$ probability. While playing the best response policy, the individual agents explore with $\epsilon_2$ probability (i.e. $\epsilon_2$ fraction of $(1-\epsilon_1)$ probability) and with the remaining probability (($1-\epsilon_2$) fraction of $(1-\epsilon_1)$) they play their best response action. Both the $\epsilon$ values are decayed exponentially. The individual agents $i$ maintain a network $Q_i(s, z, a; \theta_i)$ to approximate the best response to historical behavior of the other agents in local state $z$ when global state is $s$. The learning experience of agent $i$ is given by $<s, z, a_i, r_i, s', z'>$.

\noindent \textbf{\textit{Environment moves to the next state:}} All the individual agents observe their individual reward and update their best response values. Central agent observes the true-joint action $\bm{a}$ performed by the individual agents. Based on the true joint-action and variance ($\nu$) in the values of agents, the central agent updates its own learning. 

\noindent \textbf{\textit{Compute loss functions and optimize value and policy:}} As common with deep RL methods \cite{mnih2015human,foerster2017stabilising}, replay buffer is used to store experiences ($\mc{J}$ for the central agent and $\mc{J}_i$ for individual agent $i$) and target networks (parameterized with $\theta^{-}$) are used to increase the stability of learning. We define $\mc{L}_{\theta_\sigma}$, $\mc{L}_{\theta_\mu}$ and $\mc{L}_{\theta_i}$ as the loss functions of $\sigma$, $\mu$ and $Q_i$ networks respectively. The loss values are computed based on mini batch of experiences as follows
	{\small 
	\begin{align*}
	&\mc{L}_{\theta_{\sigma}} = \mathbb{E}_{(s,\nu,\bm{a}) \sim \mathcal{J}}\Big[ \Big( \nu  - \sigma( s, \bm{a}; \theta_\sigma)\Big)^2\Big] \\
&\mc{L}_{\theta_{\mu}} = \mathbb{E}_{(s) \sim \mathcal{J}}\Big[ -\nabla_{\theta_{\mu}} \mu(s;\theta_{\mu}) \nabla_{\bm{a}} \sigma(s, \bm{a};\theta_\sigma)|_{\bm{a} = \mu(s;\theta_{\mu})} \Big]  \\
&\mc{L}_{\theta_i} = \mathbb{E}_{(s,z,a, r, s', z') \sim \mc{J}_i}\Big[ \big(r + \gamma \cdot \max_{a'} Q_{i}(s', z', a'; \theta^{-}_i)  \\& \quad \quad \quad \quad \quad - Q_{i}(s, z, a; \theta_i) \big)^2\Big] 
	\end{align*}
	}
$\mc{L}_{\theta_\sigma}$ is computed based on mean squared error, $\mc{L}_{\theta_i}$ is computed based on TD error	\cite{sutton1988learning} and $\mc{L}_{\theta_\mu}$ is computed based on the gradient provided in Equation \ref{gradient}. Using these loss functions, policy and value functions are optimized. We have provided detailed steps of LVMQ in the appendix.

\begin{figure*}[h]
\centering
\subfloat[Real-world data \label{real-world-mean}]{\includegraphics[scale=0.24]{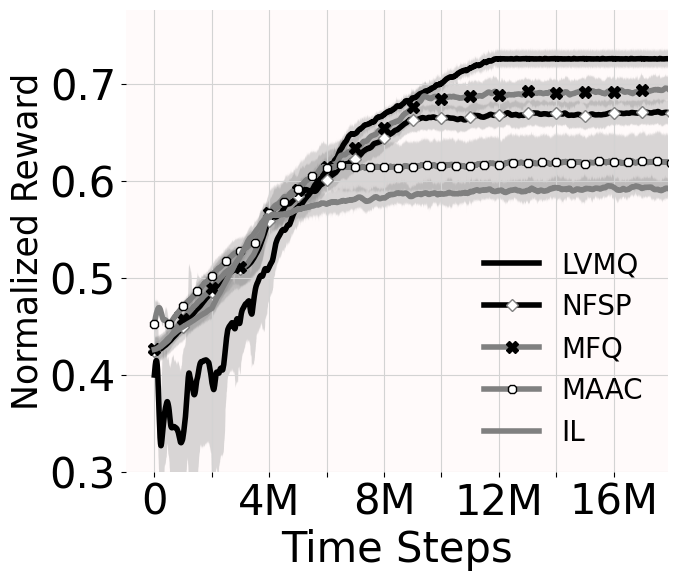}} 
\subfloat[DAR = 0.4 \label{dar04-mean}]{\includegraphics[scale=0.24]{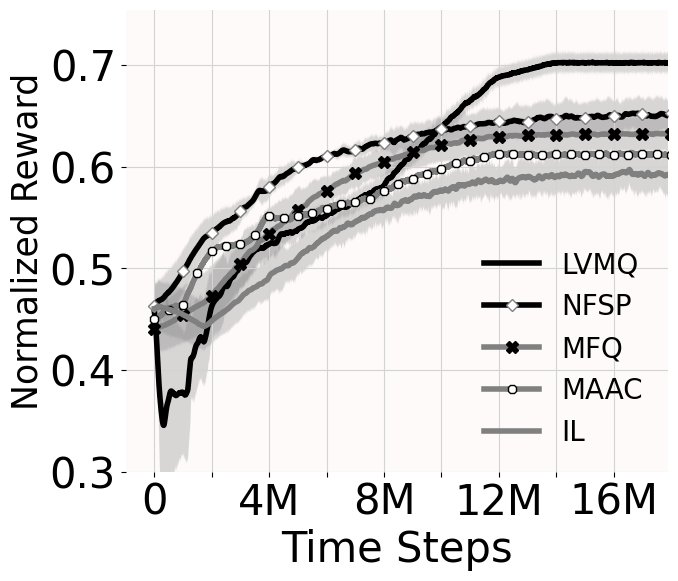}} 
\subfloat[DAR = 0.5 \label{dar05-mean}]{\includegraphics[scale=0.24]{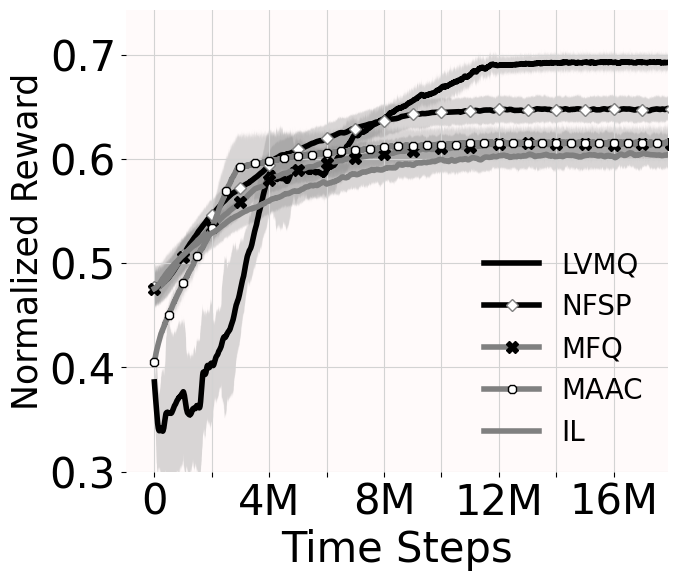}} 
\subfloat[DAR = 0.6 \label{dar06-mean}]{\includegraphics[scale=0.24]{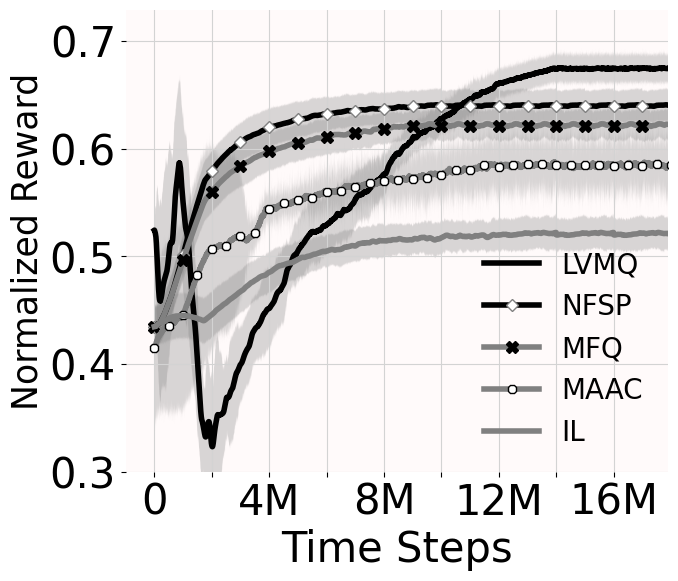}} \\ 


\subfloat[Real-world data \label{var-rw}]{\includegraphics[scale=0.24]{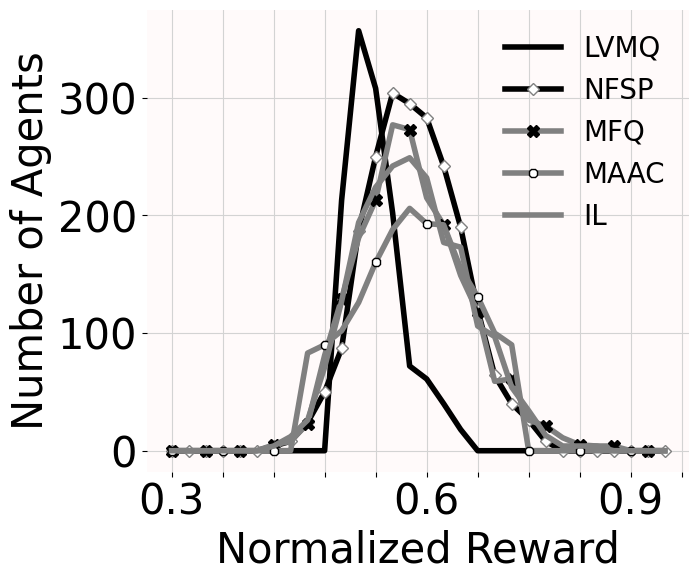}} 
\subfloat[DAR=0.4 \label{var-04}]{\includegraphics[scale=0.24]{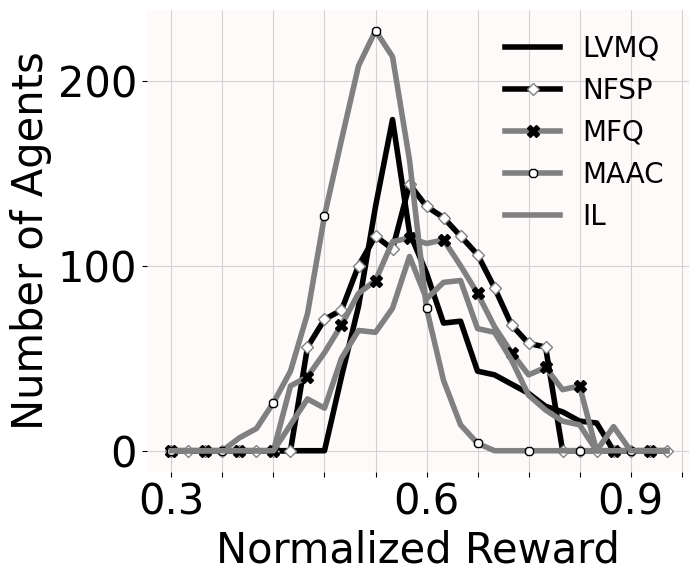}} 
\subfloat[DAR = 0.5 \label{var-05}]{\includegraphics[scale=0.24]{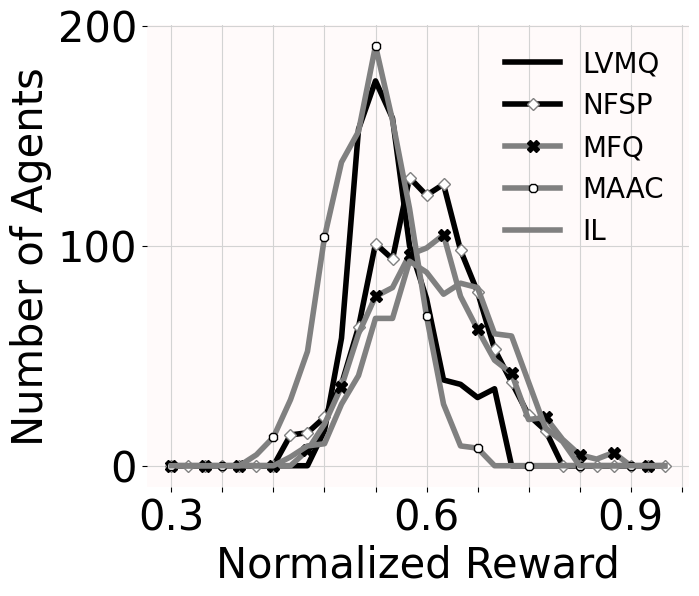}} 
\subfloat[DAR = 0.6 \label{var-06}]{\includegraphics[scale=0.24]{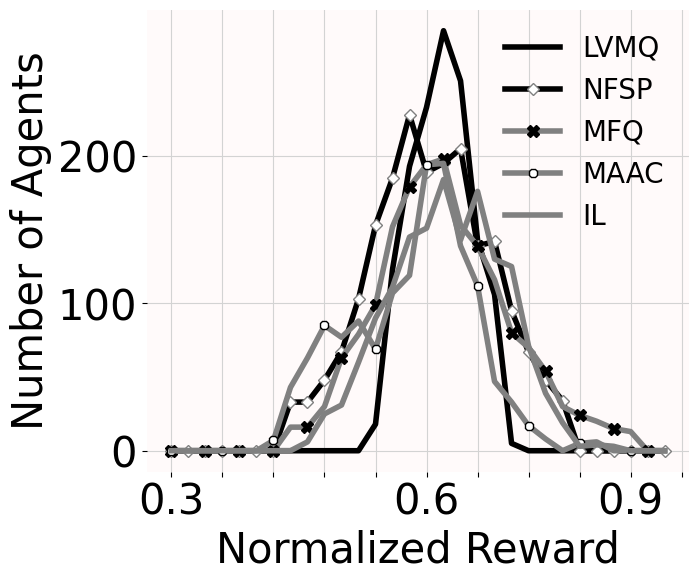}}
\caption{Mean reward (top) and distribution of values of agents in a zone (bottom) for taxi simulator using real-world and synthetic data set. Mean rewards (solid lines) are average value over 5 random seeds. Shaded regions represent one standard deviation. Value distribution is for agents in a zone with high variance in individual values.}
\label{results}
\end{figure*}
\section{Experiments}
We perform experiments on three different domains, taxi simulator based on real-world and synthetic data set \cite{verma2019entropy,verma2019correlated}, a single stage packet routing \cite{krichene2014learning} and multi-stage traffic routing \cite{wiering2000multi}. In all these domains a central agent that can assist (or provides guidance to) individual agents in approximating equilibrium policies is present. MARL strategy is suitable for these domains as these are large multi-agent systems where model based planning is not scalable. For taxi domain, we do not compare with other aggregation system related methods \cite{shah2020neural,lowalekar2018online,liu2020context} as  these methods focus on joint value and disregard the fact that the individual drivers are self-interested. Different from these methods, we focus on learning from individual drivers' perspective while considering that other drivers are also learning simultaneously.

We compare with four baseline algorithms: Independent Learner (IL), NFSP \cite{heinrich2016deep}, MFQ \cite{pmlrv80yang18d} and MAAC \cite{lowe2017multi}.  IL is a traditional Q-Learning algorithm that does not consider the actions performed by the other agents. Similar to LVMQ, MFQ and MAAC are also CTDE algorithms and they use joint action information at the time of training. However, NFSP is a self play learning algorithm and learns from individual agent's local observation. Hence, for fair comparison, we provide joint action information to NFSP as well. As mentioned by \cite{verma2019correlated}, we also observed that the original NFSP without joint action information performs worse that NFSP with joint action information. In LVMQ, the central agent learns from the collective experiences of the individual agents and suggests individuals based on this learning, whereas other algorithms receive information only about the joint action. This is one of the reasons of LVMQ performing better than other algorithms.

For decentralized learning, ideally each agent should maintain their own neural network. However, simulating thousands of learning agents at the same time requires extensive computational resources hence, we simulated 1000 agents by using 100 neural networks with 10 agents sharing a single network (providing agent id as one of the input to the network). Letting agents share a single network while providing agent id to the network is common when decentralized learning is done for a large number of agents \cite{yang2017study,pmlrv80yang18d}. 
For all the individual agents, we performed $\epsilon$-greedy exploration and it was decayed exponentially. Training was stopped once $\epsilon$ decays to 0.05. We experimented with different values of aniticipatory parameter for NFSP and 0.1 provided us the best results. The results are average over runs with 5 different random seeds.

\subsubsection{Taxi Simulator:}
We build a taxi simulator based on both real-world and synthetic data set. The data set is from a taxi company in a major Asian city. It contains information about trips (time stamp, origin/destination locations, duration, fare) and movement logs of taxi drivers (GPS locations, time stamp, driver id, status of taxi). We used methods described in \cite{verma2017augmenting} to determine the total number of zones (111) for the simulation. Each zone is considered as a local state. Then we used the real world data to compute the demand between two zones and the time taken to travel between the zones. For simulation, we computed proportional number of demand between any pair of zones based on the number of agents being used in the simulator.

We also perform experiments using synthetic data set to simulate different features such as: 
(a) Demand-to-Agent-Ratio (DAR): the average amount of demand per time step per agent; (b) trip pattern: the average length of trips can be uniform for all the zones or there can be few zones which get longer trips (non-uniform trip pattern); and (c) demand arrival rate: arrival rate of demand can either be static w.r.t. the time or it can vary with time (dynamic arrival rate).

\begin{figure*}
\centering
\subfloat[Packet Routing\label{pr}]{\includegraphics[scale=0.20]{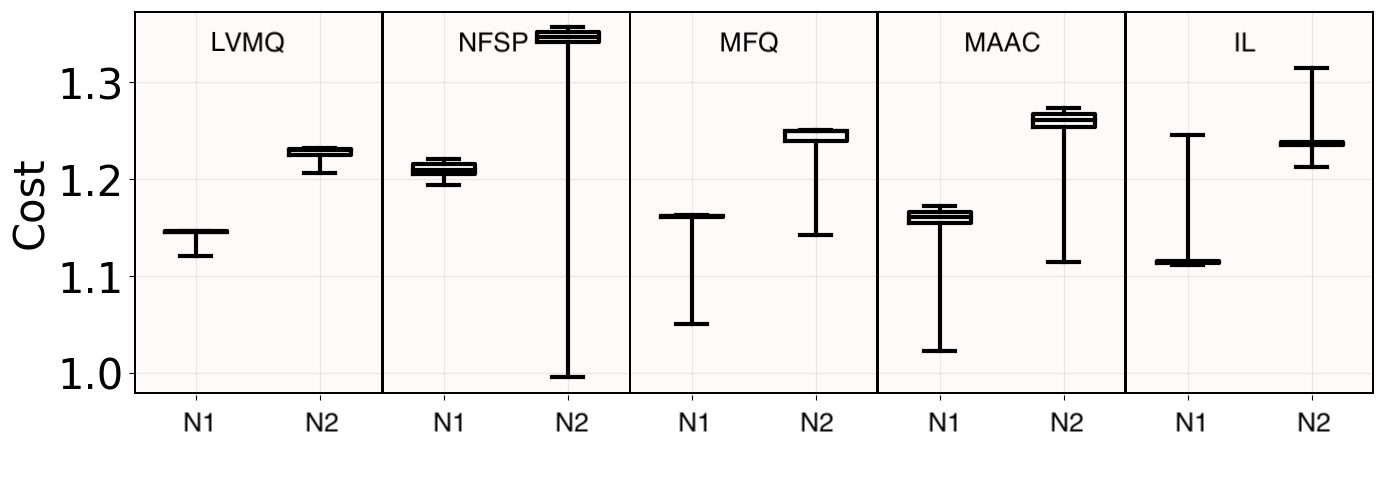}} 
\subfloat[Multi-stage Traffic Routing \label{msp}]{\includegraphics[scale=0.20]{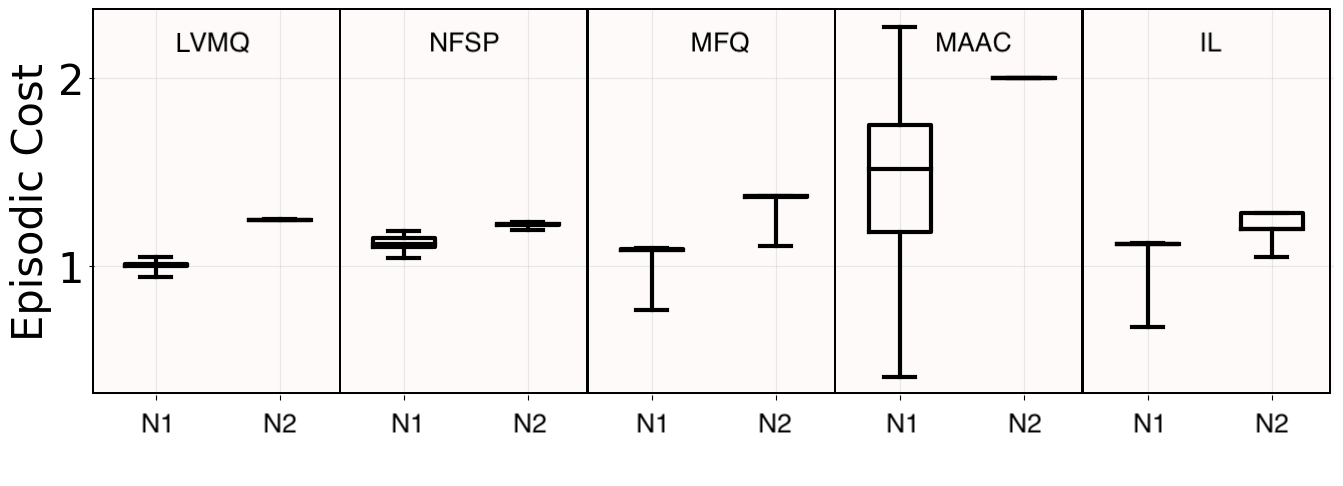}} 
\caption{Variance in values of agents for packet routing and multi-stage traffic routing example}
\label{ms-routing}
\end{figure*}
As agents try to maximize their long term revenue, we use mean reward of agents (w.r.t. the time) as the learning progresses as a comparison metric and show that LVMQ learns policy which yield higher mean values. The mean reward plots (\ref{real-world-mean}-\ref{dar06-mean}) are for the running average of mean payoff of all the agents for every $T$ (=1000) time steps. We also show variance in the values of agents after learning has converged. Plots \ref{var-rw} - \ref{var-06} show distribution of values of agents in a zone. Though the zone with highest variance in values was not same for all the 5 algorithms, the top 5 zones with highest variance were almost same for these algorithms. Here we show values of agents in the selected zone over 2$T$ time steps. We can see that the variance is minimum for LVMQ (evident from its narrower curve).

Figure \ref{real-world-mean} show that agents earn $\approx$5-10\% more value for real-world data set (an improvement of 0.5\% is considered a significant improvement \cite{xu2018large} in aggregation systems) than NFSP, MFQ and MAAC. The mean reward for LVMQ is ~3-5\% more than the second best performing algorithm. The worst performance of IL is expected as it does not receive any extra information from the centralized agent. 

Figure \ref{dar04-mean} plots mean reward of agents for a setup with dynamic arrival rate and non-uniform trip pattern with DAR=0.4. The mean reward for LVMQ is $\approx$4-10\% higher that NFSP, MFQ and MAAC. Figures \ref{dar05-mean} show results for a setup with dynamic arrival rate, uniform trip pattern and DAR=0.5. 
Comparison for an experimental setup with static arrival rate, non-uniform trip pattern and DAR=0.6 is shown in Figure \ref{dar06-mean}. 
Treating IL as a baseline, the improvement (percentage increase in mean reward) achieved by LVMQ is $\approx$10-16\%, NFSP is $\approx$5-12\%, MFQ is $\approx$1-10\% and MAAC is $\approx$1-6\%. This means that LVMQ is able to capture \textit{lost-demand} (able to serve demand before they expire). We also observed that the variance in values of individual agents was minimum for LVMQ. We discussed in Section \ref{sec-vmq} that for a joint policy to be an equilibrium policy, agents should be playing their best responses in addition to having equal values (zero variance) in same local states. As agents are playing their best response policy and variance in values is low as compared to other algorithms, LVMQ learns policies that are closer to an equilibrium policy.


\subsubsection{Packet Routing:}
To compare the performance with exact equilibrium solution, we performed experiments with a single stage (exact equilibrium values for single stage can be computed using linear programming) packet routing game \cite{krichene2014learning} explained in Section \ref{bg} (Figure \ref{routing}). The cost incurred on a path is the sum of costs on all the edges in the path. The costs functions for the edges when mass of population on the edge is $\phi$ are given by:
$c_{AB}(\phi) = \phi + 2 $, $c_{AC}(\phi) = \phi /\ 2,  c_{AD}(\phi) = \phi, 
c_{DB}(\phi) = \phi /\ 3, c_{CD}(\phi) = 3\phi, c_{EC}(\phi) = 1 /\ 2,
c_{CF}(\phi) = \phi , c_{DF}(\phi) = \phi /\ 4, c_{EF}(\phi) = \phi + 1$

If the cost functions are known, exact equilibrium policy can be computed by minimizing non-atomic version \cite{krichene2015online} of Rosenthal potential function \cite{rosenthal1973class}. We use the exact equilibrium policy and corresponding costs to compare quality of the policy learned. We also compute \textit{epsilon} (\textit{epsilon} as in epsilon-equilibrium \cite{radner1980collusive}) values of the learned policy, which is the maximum reduction in cost of an agent when it changes its policy unilaterally. Table \ref{eq-policy} compares the policies and \textit{epsilon} values where the first row contain values computed using potential minimization method. 
In the table, the policy is represented as $
((\pi^{AB}_1, \pi^{ACDB}_1,\pi^{ADB}_1), (\pi^{EF}_2, \pi^{ECDF}_2,\pi^{ECF}_2))$,
where $\pi^p_i$ is the fraction of the population of type $i$ selecting path $p$.  We see that LVMQ policy is closest to the equilibrium policy and \textit{epsilon} value is also lowest for it.

The equilibrium cost on paths computed by the potential minimization method are: $
AB = 2, ACDB = ADB = 1.14,
EF = ECDF = ECF = 1.22$. 
Figure \ref{pr} provide variance in costs of agents for both the population type. We can see that not only is the variance in the costs of agents minimum for LVMQ but the values are also very close to the values for exact equilibrium. 
We observed that other algorithms were able to perform similar to LVMQ if there is a clear choice of path based on their cost functions (the cost on one of the paths is very low such that it is advantageous for everyone to select that path). However, when complexity is increased by designing the cost function such that the choice is not very clear (see appendix), we observed that LVMQ starts outperforming other algorithms. The better performance of LVMQ is due to its advantage of receiving suggestions from a learned central agent.

{\small \begin{table}
\centering
\caption{Comparison of policies and \textit{epsilon} values}
\begin{tabular}{ | c | c | c |}
\hline 
Method & policy & \textit{epsilon} value \\ \hline
Equilibrium  &((0, 0.187, 0.813), & 0 \\Policy &(0.223, 0.053, 0.724))&\\\hline
LVMQ & ((0, 0.187, 0.813), & 0.023\\ &(0.224, 0.045, 0.731))&\\\hline
NFSP & ((0, 0.122, 0.878), &  0.342\\ &(0.005, 0.148, 0.847))&\\\hline
MFQ & ((0, 0.171, 0.829),  &0.110\\ & (0.210, 0.040, 0.750)) &\\\hline
MAAC & ((0, 0.170, 0.830),  &0.140\\& (0.201,  0.035, 0.764)) &\\\hline
IL & ((0, 0.219,  0.781), &0.129\\& (0.215, 0.048,  0.737))  &\\ \hline
\end{tabular}
\label{eq-policy}
\end{table}}

\subsubsection{Multi-Stage Traffic Routing:}
We use the same network provided in Figure \ref{routing} to depict a traffic network where two population of agents $\mc{N}_1$ and $\mc{N}_2$ navigate from node $A$ to node $B$ and from node $E$ to node $F$ respectively. Unlike to the packet routing example, agents decide about their next edge at every node. Hence it is an example of SNCG 
and the values of agents from a population at a given node would be equal at equilibrium. 

In this example, agents perform episodic learning and the episode ends when the agent reaches its destination node. The distribution of mass of population over the nodes is considered as state. Figure \ref{msp} shows that the variance in values of both the population type is minimum for LVMQ. Furthermore, we notice that for both single-stage and multi-stage cases, the 
setup is similar for agent population $\mc{N}_2$. However, for agents from $\mc{N}_1$, the values will be different from single-stage case. For example, agents selecting path $ACDB$ and $ADB$ will reach the destination node at different time steps and hence cost of agents on edge $DB$ will be different from the single-stage case. Hence we can safely assume that the equilibrium value of agents from population $\mc{N}_2$ will be 1.22 (as computed for the single-stage case) which is the value for LVMQ as shown in Figure \ref{msp}.

\section{Conclusion}
We propose a Stochastic Non-atomic Congestion Games (SNCG) model to represent anonymity in interactions and infinitesimal contribution of individual agents for aggregation systems. We show that the values of all the agents present in a local state are equal at equilibrium in SNCG. Based on this property we propose LVMQ which is a CTDE algorithm to learn approximate equilibrium policies in cases with finite yet large number of agents. Experimental results on multiple domains depict that LVMQ learns better equilibrium policies than other state-of-the-art algorithms.

\balance
\bibliography{vmq_main}
\newpage
\setcounter{prop}{0}
\onecolumn
\section*{Supplementary Material}
\section{Properties of SNCG}
In this section, we show that values of agents in the same local state have either equal or close to equal values at equilibrium. We begin with the case of non-atomic agents and then move to the case with large number of agents and where mass of an agent is non-zero. 

\subsection{Non-Atomic Agents}
In case of non-atomic agents, we first prove that values of other agents do not change if only one agent changes its policy (Proposition~\ref{agent-value}). This property is then later used to prove that values of agents present in a local state are equal at equilibrium (Proposition~\ref{equal-value}).

\begin{prop} 
Values of other agents do not change if agent $i$ alone changes its policy from $\pi_i$ to $\pi'_i$, i.e., for any agent $j$ in any local state $z$:
$ v_{jz} (s, \pi_{jz}, \bm{\pi}_{-j} ) = v_{jz} (s, \pi_{jz}, \bm{\pi}'_{-j} )
\text{ where }  \bm{\pi}_{-j} = \big(\pi_i, (\pi_k)_{k \in {\mc{N} \setminus \{i, j\}}}\big) $ and $\bm{\pi}'_{-j} = \big(\pi'_i, (\pi_k)_{k \in {\mc{N} \setminus \{i, j\}}}\big)$
\end{prop}
\begin{proof}

As can be observed in the value function equation (Equation 1 in the main paper), change of policy for agent $i$ can have impact on agent $j$'s value, $v_{jz}(s, \pi_{jz}, \bm{\pi}_{-j})$, due to the following three components: \\
\noindent immediate reward, i.e., ${\cal R}_z(s,\phi^{\pi_{jz}(s)}(s,\mathbf{a}))$ ; or\\
\noindent transition function, i.e., ${\cal T}(s'|s,\bm{a})$; or \\
\noindent future value, $v_{jz'}(s',\pi_{jz'}, \bm{\pi}_{-j})$

All  three components are dependent on one common summary factor, which is the mass of agents taking a certain action ${a}$ in state $s$ with zone $z$.  We will refer to this as $\tilde{f}_{{z}}^{a}(s)$ (we are using $\tilde{f}$ to indicate that it is a change  from original number of agents, $f$). Due to change in policy of $i$, the number of agents taking an action ${a}$ in a zone ${z}$ can either increase or decrease. Without loss of generality, let us assume agent $i$ is using a deterministic policy\footnote{If agent $i$ employs a stochastic policy, the mass argument used in this proof is still applicable.}.

If change in policy moves agent $i$ away from zone ${z}$, then the new mass of agents taking action ${a}$ in zone $z$ is  $$\tilde{f}^{{a}}_{{z}}(s) = \int_{k \in \mc{N}^s_z \setminus {i}} \mathbbm{1}_{(act(k)={a})} dm(k)$$

Since $f$ is primarily mass of agents, which is a Lebesgue measure and using the \textit{countable additivity} property of Lebesgue measure \cite{bogachev2007measure,hartman2014theory}, we have:
{\small  \begin{align*}
 \tilde{f}^{a}_{z}(s)  &= \int_{k \in {\mc{N}^s_z}} \mathbbm{1}_{(act(k)=a)} dm(k) -   \int_{k \in \{i\}} \mathbbm{1}_{(act(k)=a)} dm(k)
 \end{align*}
 }
 Similarly if change in agent $i$ policy implies moving into zone ${z}$, the new value of number of agents taking action ${a}$ will be:
{\small  \begin{align*}
 \tilde{f}^{a}_{z}(s)  &= \int_{k \in {\mc{N}^s_z}} \mathbbm{1}_{(act(k)=a)} dm(k) +   \int_{k \in \{i\}} \mathbbm{1}_{(act(k)=a)} dm(k)
 \end{align*}
 }

\textup{Since integral at a point in continuous space is 0 and mass measure is non-atomic, so $\{i\}$ is a null set and $m(\{i\})=0$, and hence, 
 $\tilde{f}^{a}_{z}(s)= \int_{k \in {\mc{Z}^s_z}} \mathbbm{1}_{(act(k)=a)} dm(k)$.
Since $f_{z}^{a}(s) = \tilde{f}^{a}_{z}(s)$ for any given $z$ and $a$, none of the three components (immediate reward, transition function and future value) change. } 
\end{proof}

Based on Proposition~\ref{agent-value}, we now show that at Nash Equilibrium for SNCG, values of agents that are in same local state are equal.  A joint policy $\bm{\pi}$ is a Nash equilibrium if for all $z \in \mc{Z}$ and for all $i \in \mc{N}^s_z$, there is no incentive for anyone to deviate unilaterally, i.e. $
v_{iz}(s, \pi_{iz}, \bm{\pi}_{-i}) \ge v_{iz}(s, \pi'_{iz}, \bm{\pi}_{-i}) \forall s \in \mc{S}, \forall i \in \mc{N}^s_z, \forall z \in \mc{Z}, \forall \pi_{iz}, \pi'_{iz} \in \Pi_z 
$

\begin{prop} 
Values of agents present in a local state are equal at equilibrium, i.e.,
{\small \begin{align}
v_{iz}(s, \pi^*_{iz}, \bm{\pi}^*_{-i}) = v_{jz}(s, \pi^*_{jz}, \bm{\pi}^*_{-j}), \forall s \in \mc{S}, \forall i,j \in \mc{N}^s_z, \nonumber \\ \forall z \in \mc{Z}, \forall \pi^*_{iz}(s), \pi^*_{jz}(s) \in \Pi_z
\end{align}}
where superscript * denotes equilibrium policies.
\end{prop}
\begin{proof}
In the proof of Proposition~\ref{agent-value}, we showed that adding or subtracting one agent from a local state does not change other agent's values, as contribution of one agent is infinitesimal. Thus,
\begin{align}
v_{iz}(s, \pi_{iz}^{*}, \bm{\pi}^{*}_{-i}) &= v_{iz}(s, \pi_{iz}^{*}, \bm{\pi}^{*}) \textbf{ and also } \nonumber \\
v_{iz}(s, \pi_{iz}^{'*}, \bm{\pi}_{-i}) &= v_{iz}(s, \pi_{iz}^{'*}, \bm{\pi}^{*}) 
\end{align}
From equilibrium condition (value of equilibrium strategy is greater than or equal to value of any other strategy for each agent in every state) employed in NashQ learning~\cite{hu2003nash} and the above results from Proposition~\ref{agent-value}, we have (accounting for local states):
\begin{align}
v_{iz}(s, \pi^{*}_{iz}, \bm{\pi}^*) &\geq v_{iz}(s, \pi_{iz}, \bm{\pi}^*), \forall \pi_{iz} \neq \pi^{*}_{iz} \nonumber\\
v_{jz}(s, \pi^{*}_{jz}, \bm{\pi}^*) &\geq v_{jz}(s, \pi_{jz}, \bm{\pi}^*), \forall \pi_{jz} \neq \pi^{*}_{jz} \nonumber
\end{align}
Since the set of policies available to each agent are the same in each local state, i.e. one of the possible values of $\pi_{iz}$ is $\pi^*_{jz}$ and similarly one of the possible values of $\pi_{jz}$ is $\pi^*_{iz}$. Hence
\begin{align}
v_{iz}(s, \pi^{*}_{iz}, \bm{\pi}^*) &\geq v_{iz}(s, \pi^*_{jz}, \bm{\pi}^*) \nonumber\\
v_{jz}(s, \pi^{*}_{jz}, \bm{\pi}^*) &\geq v_{jz}(s, \pi^*_{iz}, \bm{\pi}^*) \nonumber
\end{align}
Since agents do not have specific identity (i.e., rewards are dependent on zone, transitions are on joint state), $i$ and $j$ are interchangeable. 
The above set of equations are only feasible if agent $i$ and $j$ in same local state have equal values, i.e.,
\begin{align}
v_{iz}(s, \pi^{*}_{iz}, \bm{\pi}^*) &= v_{jz}(s, \pi^{*}_{jz}, \bm{\pi}^*) \nonumber\\
\intertext{Therefore:}
v_{iz}(s, \pi^{*}_{iz}, \bm{\pi}^*_{-i}) &= v_{jz}(s, \pi^{*}_{jz}, \bm{\pi}^*_{-j}) \nonumber
\end{align}
\end{proof}
\begin{corollary}
For problems with multiple types of agents, values of same type of agents are equal in a local state at equilibrium for non-atomic case. 
\end{corollary}

In non-atomic case, individual agents have zero mass and we have shown here that values of agents with same local states will be equal at equilibrium. We now move  onto domains with large number of agents but not completely non-atomic. Since agents have non-zero mass, the proofs above do not hold. 

\subsection{Nearly Non-atomic}
In aggregation systems, where there are large number of agents but each agent has a small mass, we provide proofs for two reward cases separately. 
\begin{prop}
Consider two agents, $i$ and $j$ in zone $z$, and let $\bm{\pi}^* = (\pi_1^*,\cdots,\pi_i^*, \cdots, \pi_j^*, \cdots)$ be the equilibrium policy. For \textbf{R1} setting, we have:
$$ v_{iz}(s, \pi_{i}^*, \bm{\pi}^*_{-i}) = v_{jz}(s, \pi_{j}^*, \bm{\pi}^*_{-j})$$
\end{prop}
\begin{proof}
We show this using mathematical induction on the horizon, $t$. 

\noindent \textbf{For t = 1}:
$$v_{iz}^t(s, \pi_i^*, \bm{\pi}_{-i}^*) = R_z(s, \pi(s,\bm{a}))$$
$$v_{jz}^t(s, \pi_j^*, \bm{\pi}_{-j}^*) = R_z(s, \pi(s,\tilde{\bm{a}}))$$
It should be noted that
$$(\pi_i^*, \bm{\pi}_{-i}^*) = (\pi_j^*, \bm{\pi}_{-j}^*)$$
Therefore, $\phi(s,\bm{a}) = \phi(s, \tilde{\bm{a}})$ and hence 
$$v_{iz}^t(s, \pi_{i}^*, \bm{\pi}^*_{-i}) = v_{jz}^t(s, \pi_{j}^*, \bm{\pi}^*_{-j}), \forall s$$ 
\end{proof}

Let us assume that it holds for $t = m$, i.e.,
$$v_{iz}^t(s, \pi_{i}^*, \bm{\pi}_{-i}^*)) = v_{jz}^t(s, \pi_{j}^*, \bm{\pi}_{-j}^*)) , \forall s$$

For $t = m + 1$, we have
\begin{align} 
v_{iz}^{t+1}(s, \pi_{i}^*, \bm{\pi}_{-i}^*) &= R_{z}(s, \phi(s,\bm{a})) + \gamma \sum_{s'} {\cal T}(s'|s,\bm{a}) v_{iz}^t(s, \pi_{i}^*, \bm{\pi}_{-i}^*) \nonumber\\
v_{jz}^{t+1}(s, \pi_{j}^*, \bm{\pi}_{-j}^*) &= R_{z}(s, \phi(s,\tilde{\bm{a}})) + \gamma \sum_{s'} {\cal T}(s'|s,\tilde{\bm{a}}) v_{jz}^t(s, \pi_{j}^*, \bm{\pi}_{-j}^*) \nonumber
\intertext{Similar to t = 1, we have $(\pi_i^*, \bm{\pi}_{-i}^*) = (\pi_j^*, \bm{\pi}_{-j}^*)$ and hence $\bm{a} = \tilde{\bm{a}}$ and $\phi(s,\bm{a}) = \phi(s, \tilde{\bm{a}})$. Given this and our assumption above, we have: }
v_{iz}^{m+1}(s, \pi_{i}^*, \bm{\pi}^*_{-i}) &= v_{jz}^{m+1}(s, \pi_{j}^*, \bm{\pi}^*_{-j})
\end{align}
For infinite horizon problems, we can just write this as 
$$ v_{iz}(s, \pi_{i}^*, \bm{\pi}^*_{-i}) = v_{jz}(s, \pi_{j}^*, \bm{\pi}^*_{-j}) $$

$\hfill \qedsymbol$

\noindent For reward case, \textbf{R2} and nearly non-atomic agents, we are unable to provide a proof of value equivalence for agents in same local state. However, we can provide an upper bound on the difference value of any two agents in a given zone $z$. 

Let assume that for nearly non-atomic case, the maximum impact an agent can have in terms of reward at a time step is bounded by $\epsilon$, i.e., 
\begin{align}
\max_{a,b} \left[ R_z(s, \phi^{a}(s,\bm{a})) - R_z(s, \phi^b(s, \bm{a})) \right] \leq \epsilon \quad \quad \forall a,b \in \mc{A}_z, \forall s \in \mc{S}, \forall z \in \mc{Z}, \forall \bm{a} \in \mc{A} \label{eq:nearlynonatomic}
\end{align}
This is a fair assumption, since there are a large number of agents in the nearly non-atomic case and each agent has a minimal impact. 

For an equilibrium policy,$\pi^* = (\pi_1^*, \ldots,\pi_i^*, \ldots, \pi_j^*, \ldots)$, the difference in values of any two agents, $i$ and $j$ in the same zone $z$ is given 
\begin{align}
    v_{iz}(s, \pi_{i}^*, \bm{\pi}^*_{-i}) - v_{jz}(s, \pi_{j}^*, \bm{\pi}^*_{-j}) &= R_z(s, \phi^{\pi_{iz}^*(s)}(s, \bm{a})) - R_z(s, \phi^{\pi_{jz}^*(s)}(s, \bm{a})) \nonumber \\ &+ \gamma \sum_{s'} {\cal T}(s'|s,\bm{a}) \left [ v_{iz'}(s', \pi_{i}^*, \bm{\pi}^*_{-i}) - v_{jz'}(s', \pi_{j}^*, \bm{\pi}^*_{-j}) \right] \nonumber\\
\intertext{It should be noted that $(\pi_i^*, \bm{\pi}_{-i}^*) = (\pi_j^*, \bm{\pi}_{-j}^*)$ and therefore the transition function values, ${\cal T}(s'|s,a)$ will be the same for both agents. Next we substitute Equation~\ref{eq:nearlynonatomic} and use the sum of geometric progression to obtain (for infinite horizon case)}
&v_{iz}(s, \pi_{i}^*, \bm{\pi}^*_{-i}) - v_{jz}(s, \pi_{j}^*, \bm{\pi}^*_{-j}) \leq \frac{\epsilon}{1 - \gamma}
\end{align}
For the finite horizon case, it will be $T*\epsilon$ (assuming there is no discount factor).

\section{Local Variance minimizing Multi-agent Q-learning, LVMQ}
Algorithm \ref{algo-VMQ} provides detailed steps of LVMQ.

\begin{algorithm} [h]
\caption{LVMQ}
\label{algo-VMQ}
\begin{algorithmic}[1]
	\STATE {Initialize replay buffer $\mc{J}$, action-variance network $\sigma(s, \bm{a};\theta_\sigma)$, policy network $\mu(s;\theta_{\mu})$}
\STATE {Initialize replay buffer $\mc{J}_i$, action-value network $Q_{i}(s,z,a; \theta_i)$ and corresponding target network with parameter $\theta^-_i$ for all the individual agents $i$}
\WHILE{not converged}
\FOR{ $z \in \mc{Z}$}
	\FOR{$i \in \mc{N}^s_z$}
	\STATE{compute value of $i$, $v_{iz} = \underset{a}{\text{max}} Q_{i}(s,z,a;\theta_i)$}
	\ENDFOR
\STATE{Compute $\nu_z$, variance in $v_{iz}$ values for $i \in \mc{N}^s_z$}
\ENDFOR
\STATE{Compute mean variance $\nu = \dfrac{1}{|\mc{Z}|} \sum_{z \in \mc{Z}} \nu_z$}
\STATE{Compute suggested joint action by the central entity $\bm{a}^c \leftarrow \mu(s,\theta^{\mu})$.\label{suggested_act}}
	\FOR {all $z \in \mc{Z}$ and for all agent $i \in \mc{N}^s_z$ }
	\STATE{with probability $\epsilon_1$, \\ 
	$\quad \quad a_i \leftarrow \text{sample from }\bm{a}^c_z $ \\
	with remaining probability $1-\epsilon_1$ \\
	$\quad \quad a_i \leftarrow  \epsilon_2\text{-greedy} (Q_{i})  $\label{indi-exploration}}
		\STATE{Perform action $a_i$ and observe immediate reward $r_i$ and next local state $z'$}
	\ENDFOR
	\STATE{Compute true joint action $\bm{a}$ and observe next state $s'$}
	 \STATE{Store transition $(s, \nu, \bm{a})$ in  $\mc{J}$ and respective transitions $(s, z_i, a_i, r_i, s', z'_i)$ in $\mc{J}_i$ for all agents $i$}
	\STATE{Periodically update the network parameters by minimizing the loss functions 
	\begin{align*}
	&\mc{L}_{\theta_{\sigma}} = \mathbb{E}_{(s,\nu,\bm{a}) \sim \mathcal{J}}\Big[ \Big( \nu  - \sigma( s, \bm{a}; \theta_\sigma)\Big)^2\Big]   \\
&\mc{L}_{\theta_{\mu}} = \mathbb{E}_{(s) \sim \mathcal{J}}\Big[ -\nabla_{\theta_{\mu}} \mu(s;\theta_{\mu}) \nabla_{\bm{a}} \sigma(s, \bm{a};\theta_\sigma)|_{\bm{a} = \mu(s;\theta_{\mu})} \Big]  \\
&\mc{L}_{\theta_i} = \mathbb{E}_{(s,z,a, r, s', z') \sim \mc{J}_i}\Big[ \big(r + \gamma \cdot \max_{a'} Q_{i}(s', z', a'; \theta^{-}_i) \nonumber - Q_{i}(s, z, a; \theta_i) \big)^2\Big] 
\end{align*}	 }
\STATE{Periodically update the target network parameters}
\ENDWHILE
\end{algorithmic}
\end{algorithm} 

\section{Hyperparameters}
Our neural network consists of one hidden layer with 256 nodes. We also use dropout layer and layer normalization to prevent the network from overfitting. 
We ran our experiments on a 64-bit Ubuntu machine with 256G memory and used tensorflow for deep learning framework. We performed hyperparameter search for learning rate (1e-3 - 1e-5) and anticipatory parameter (0.1 - 0.8 for NFSP). The dropout rate was set to 50\%. We used Adam optimizer for all the experimental domains. Learning rate was set to 1e-5 for all the experiments. Our learning rate is smaller than the rates generally used. We believe having low learning rate helps in dealing with non-stationarity when there are large number of agents present.
We used 0, 1, 2022, 5e3, 1e5 as random seeds for 5 different runs of experiment. 

\section{Description of Simulator}
The real-world data set is from a taxi company in a major Asian city which we will share after the paper is published. Demands (for both real-world and synthetic data set) are generated with a \textit{time-to-live} value and the demand expires if it is not served within this time periods. Agents start at random zones and based on their learned policy move across the zones if no demand has been assigned to them. At every time step, the simulator assigns a trip probabilistically to the agents based on the number of agents present at the zone and the customer demand in that zone. Once a trip is assigned to an agent, it follows a fixed path to the destination and is no longer considered a candidate for a new trip. The agent starts following the learned policy only after it has completed the trip and is again ready to serve new trips. The fare of the trip minus the cost of traveling is used as the immediate reward for actions. The learning is episodic and an episode ends after every $T$ time steps.

\section{Discussion on number of agents}
LVMQ is motivated from the equilibrium property of having same value (i.e. zero variance) in a local state of SNCG. However, SNCG considers agents to be non-atomic. To check how many number of agents are required to consider an agent to be nearly non-atomic, we performed ablation study with different number of agents. Figure \ref{ablation} provides results when we simulated the taxi domain using real world dataset with 20, 100, 200 and 1000 agents. Though the average reward is still higher for LVMQ, the variance and standard deviation is relatively high when there are 20 agents (Figures \ref{compare-20} and \ref{var-20}). The values improve when 100 agents are used. It is evident from Figures \ref{compare-200}, \ref{var-200}, \ref{compare-1000} and \ref{var-1000} that LVMQ performs well when number of agent is around 200. Hence, we feel that we can approximate non-atomic behaviour when total number of agents are more than 200. 

\begin{figure*}[h]
\centering
\subfloat[20 agents \label{compare-20}]{\includegraphics[scale=0.2]{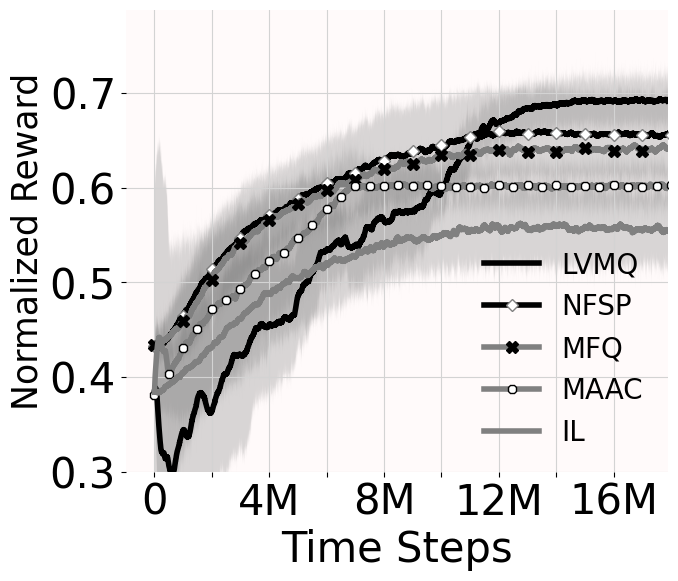}} 
\subfloat[100 agents \label{compare-100}]{\includegraphics[scale=0.2]{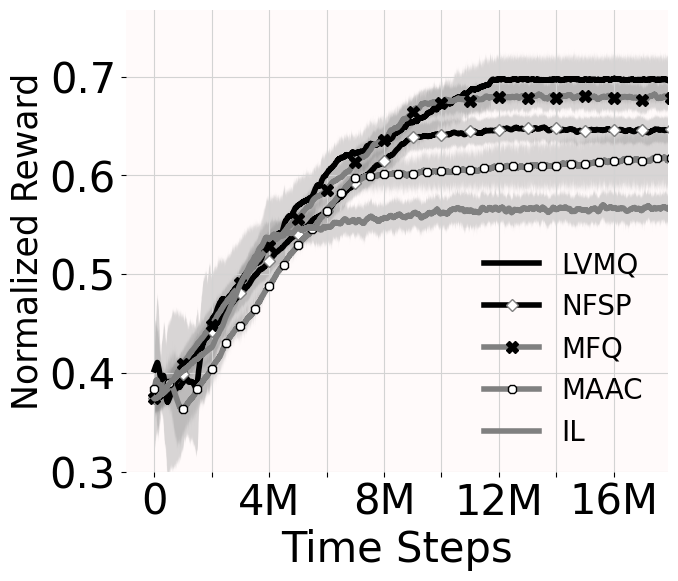}} 
\subfloat[200 agents \label{compare-200}]{\includegraphics[scale=0.2]{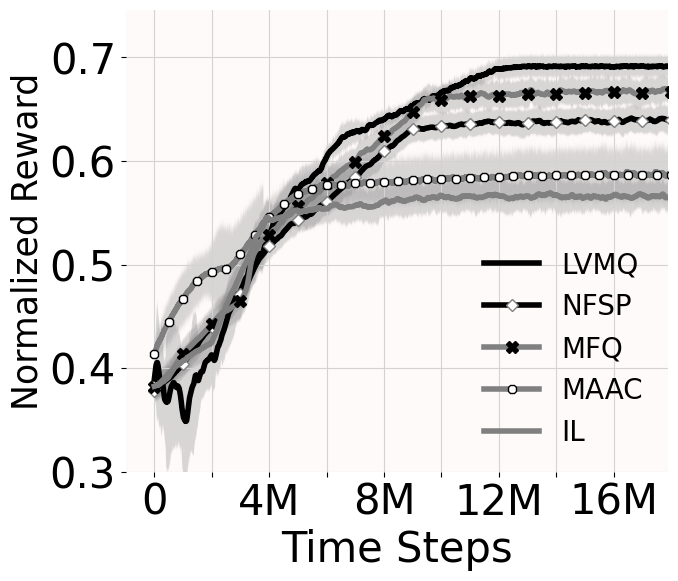}} 
\subfloat[1000 agents \label{compare-1000}]{\includegraphics[scale=0.2]{compare-rw.png}}\\
\subfloat[20 agents \label{var-20}]{\includegraphics[scale=0.2]{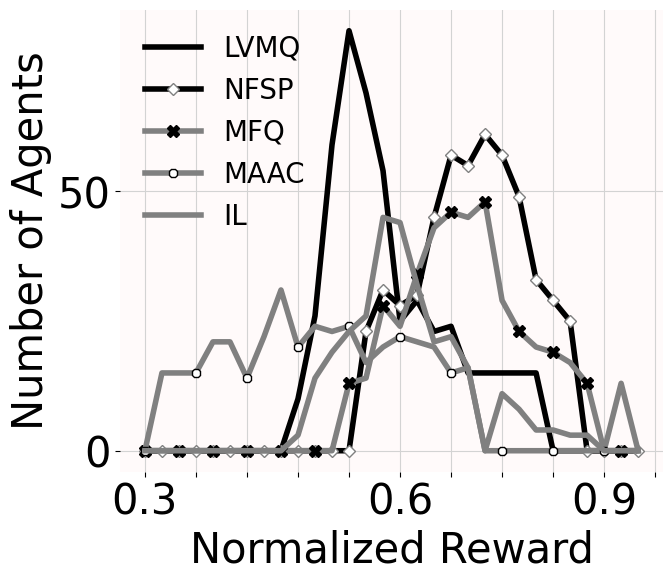}} 
\subfloat[100 agents \label{var-100}]{\includegraphics[scale=0.2]{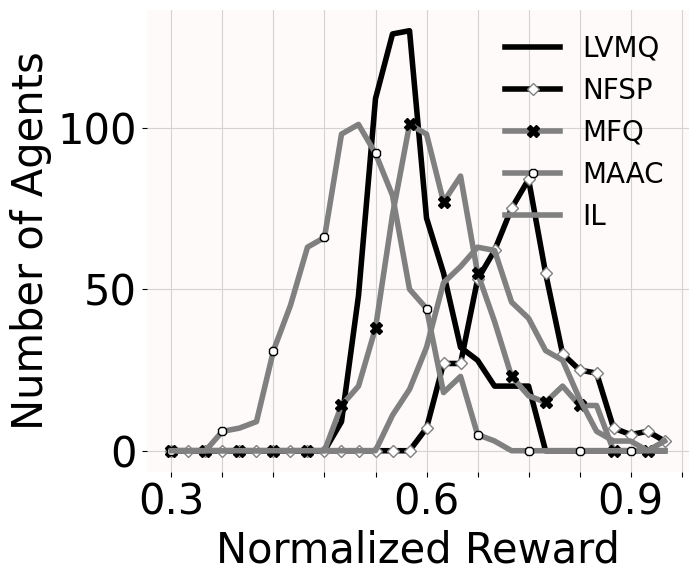}} 
\subfloat[200 agents \label{var-200}]{\includegraphics[scale=0.2]{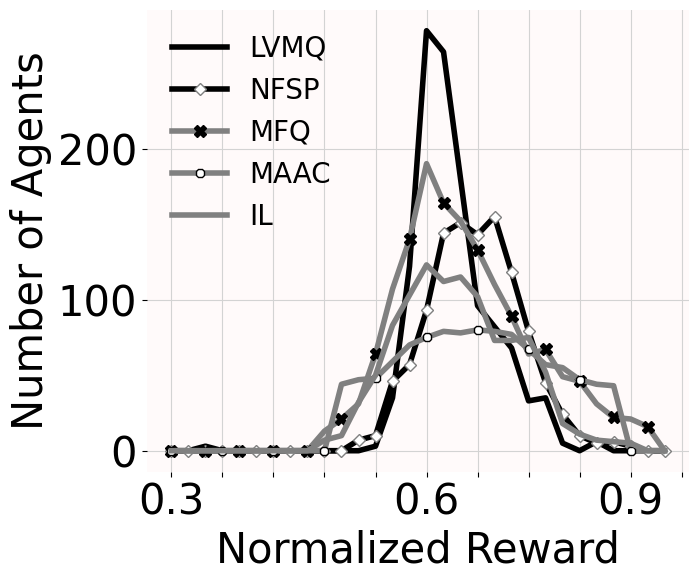}} 
\subfloat[1000 agents \label{var-1000}]{\includegraphics[scale=0.2]{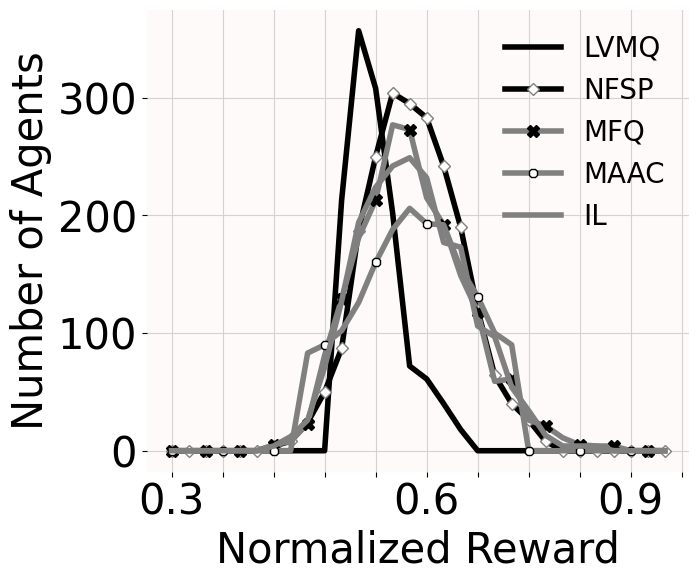}}
\caption{Performance of algorithms when 20, 100, 200 and 1000 agents are used for experiments. }
\label{ablation}
\end{figure*}

\section{Packet Routing}
\begin{figure}
  \centering
    \includegraphics[scale=0.35]{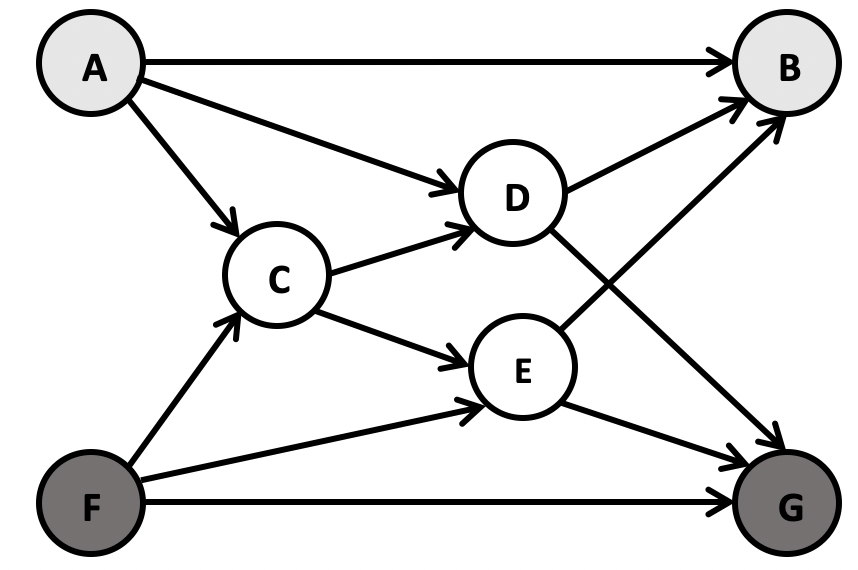}
  \caption{Routing network}
  \label{routing-new}
\end{figure}
Figure \ref{routing-new} shows another example of packet routing domain where two different populations of agents $\mc{N}_1$ and $\mc{N}_2$ of mass 0.5 each share the routing network. Population $\mc{N}_1$ sends packets from node A to node B and population $\mc{N}_2$ sends from node F to node G.  AB, ADB, ACDB and ACEB are the paths available to population $\mc{N}_1$, whereas paths FG, FEG, FCEG and FCDG are available to population $\mc{N}_2$. 

We now explain two examples where we used different cost functions to vary the complexity of learning. For example 1, the costs functions are such that both the populations do not share any edge at equilibrium. For the second example, the costs functions are selected in a way such that all the edges have non-zero mass of agents at equilibrium.
\subsubsection*{Example 1}
\begin{figure}
  \centering
    \includegraphics[scale=0.32]{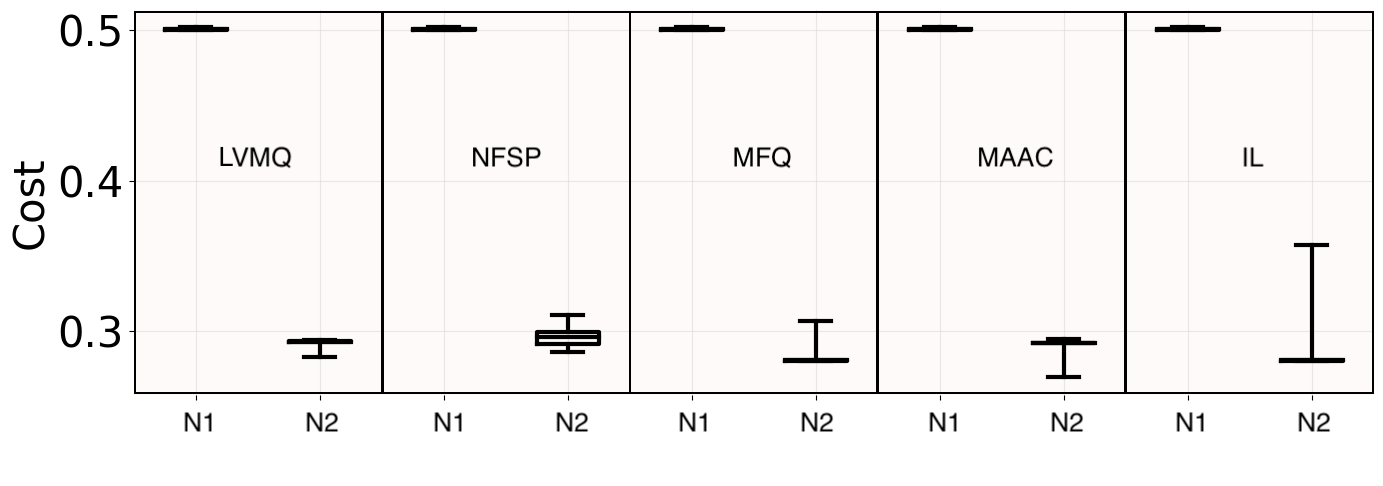}
  \caption{Variance in individual costs for example 1}
  \label{routing1}
\end{figure}
The costs functions are as follows
\begin{align*}
c_{AB}(\phi) &= \dfrac{\phi}{2} &c_{AC}(\phi) &= 1  &c_{AD}(\phi) &= \dfrac{\phi}{3} &c_{CD}(\phi) &= \dfrac{\phi}{4} \\ 
c_{CE}(\phi) &= \dfrac{1}{2}  &c_{DB}(\phi) &= \phi +1  &c_{DG}(\phi) &= \dfrac{\phi}{6} &c_{EB}(\phi)& = \dfrac{\phi}{5}\\
c_{EG}(\phi) &= \dfrac{\phi}{5} &c_{FC}(\phi) &= 2 \phi  &c_{FE}(\phi) &= \phi + \dfrac{1}{2} &c_{FG}(\phi) &= \dfrac{\phi}{3} 
\end{align*}
Please notice that the cost function of AB is significantly lower than the other paths available to population $\mc{N}_1$. We use $
((\pi^{AB}_1, \pi^{ADB}_1,\pi^{ACDB}_1, \pi^{ACEB}_1), (\pi^{EG}_2, \pi^{EFG}_2,\pi^{FCEG}_2, \pi^{FCDG}_2))$ to represent the policy where $\pi^p_i$ is the fraction of the population of type $i$ selecting path $p$. The exact equilibrium computed by solving LP (minimizing the potential function) is $((1. 0,0,0),(0.879,0,0,0.121))$ and the corresponding costs on each paths are as follows
$$
AB = 0.5, ADB = 1.0, ACDB = 2.0, ACEB=1.5,
FG=0.29, FEG=0.5, FCEG=0.74, FCDG=0.29
$$
We can see that costs on paths with zero mass (ex. ADB, ACDB, ACEB) are higher than the costs on paths with non-zero mass (AB) of agents. Also, costs on path with non-zero mass (ex. FG, FCDG) are equal at equilibrium.
\begin{table} [h]
\centering
\caption{Comparison of \textit{epsilon} values}
\begin{tabular}{ | c | c | c |}
\hline 
Method & \textit{epsilon} value for example 1 & \textit{epsilon} value for example 2 \\ \hline
LVMQ & 0.01 & 0.07 \\ \hline
NFSP & 0.02 &0.44\\ \hline
MFQ & 0.46 & 0.47 \\  \hline
MAAC & 0.03 & 0.57 \\  \hline
IL & 0.11 & 0.79\\  \hline
\end{tabular}
\label{eq-policy-sup}
\end{table}
\subsubsection*{Example 2}

Following are the cost functions
\begin{align*}
c_{AB}(\phi) &= \phi + \dfrac{5}{2} &c_{AC}(\phi) &= \dfrac{\phi}{100} + 1  &c_{AD}(\phi) &= \dfrac{\phi}{3} + 1 &c_{CD}(\phi) &= \dfrac{\phi}{100} \\ 
c_{CE}(\phi) &= \dfrac{\phi}{10} + \dfrac{2}{5}  &c_{DB}(\phi) &= \phi + 1 &c_{DG}(\phi) &= \dfrac{\phi}{100} + \dfrac{3}{20}  &c_{EB}(\phi)& = \dfrac{6\phi}{5} + 1\\
c_{EG}(\phi) &= \dfrac{\phi}{100} &c_{FC}(\phi) &= 5\phi + 2 &c_{FE}(\phi) &= 3\phi + \dfrac{7}{2} &c_{FG}(\phi) &= 5\phi + 3
\end{align*}
The exact equilibrium computed for this example is $((0.161, 0.063, 0.578, 0.198), (0.273, 0.287, 0, 0.44))$ and corresponding costs on each path are
$$
AB = ADB = ACDB = ACEB= 2.66,
FG=FEG= FCDG=4.36, FCEG=4.63
$$
\begin{figure}
  \centering
    \includegraphics[scale=0.32]{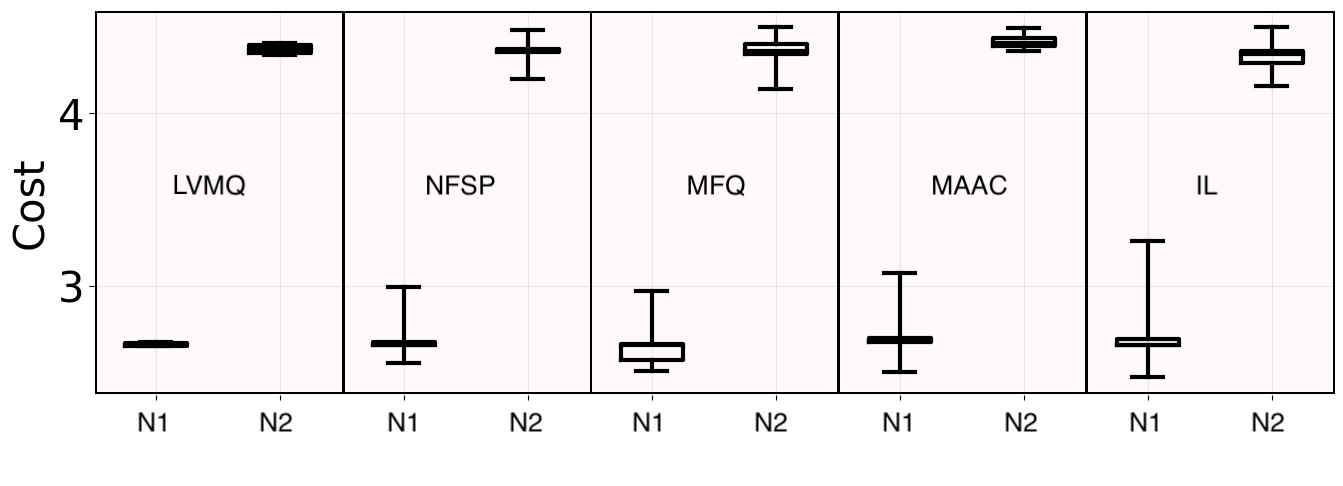}
  \caption{Variance in individual costs for example 2}
  \label{routing2}
\end{figure}
Figures \ref{routing1} and \ref{routing2} show variance in individual agents for examples 1 and 2 respectively. On the $x-$axis, $N_i$ means the variance is for population type $i$ for the corresponding algorithm. Table \ref{eq-policy-sup} contain the \textit{epsilon} values (the maximum reduction in the cost of an agent when it changes its policy unilaterally) for both the examples. 

We observe that all the algorithms learn near equilibrium policy for example 1 (specifically the policy of population $\mc{N}_1$) and \textit{epsilon} values are also low. However, LVMQ performs better than other algorithms for the second example. The variance is minimum for LVMQ and the \textit{epsilon} value is also very low for it. Hence, LVMQ is able to learn policies which are close to equilibrium when complexity of learning is high.

\section{Limitation}
One limitation of LVMQ is its inability to distinguish between good equilibrium policy (high value) and bad equilibrium policy (low value) when multiple equilibria exists. Hence, LVMQ might converge towards the low value equilibrium policy if there are multiple equilibria. 

From aggregation domain perspective, another limitation is that the work does not consider dynamic population of agents. The assumption is that the mass of agents remains $1$ trough out the learning episode ($T$ time steps). However, in reality agents (taxi drivers) keep leaving/joining the environment and this may affect the learning performance if there is a significant rise/drop in overall mass of agents. We also do not consider heterogeneous agents and assume that utilities of all the agents are uniform. In reality, utility of same trip (though immediate reward is same) might be different for different agents.

Above mentioned limitations are separate research topics and we will explore these directions in our future work.

\end{document}